\documentclass{ifacconf}

\usepackage{graphicx}
\usepackage{natbib}

\usepackage{amsmath, amscd, amsfonts, amssymb}
\usepackage{enumitem}
\newcommand{\mc}[1]{\mathcal{#1}}
\newcommand{\mb}[1]{\mathbb{#1}}

\DeclareMathOperator{\cl}{cl}
\newcommand{\orb}[3]{\{{#1}^{#2 k}(#3)\}}
\newcommand{\seq}[3]{({#1}^{#2 k}(#3))}

\begin{document}
\begin{frontmatter}

\title{On the Nonexistence of Continuous Immersions for Discrete-time Systems}
%\thanksref{sponsors}  for thanking sponsors
\thanks[sponsors]{This work is supported in part by ONR CLEVR-AI MURI (N00014-21-1-2431) and NSF grant CNS \# 2434331.}

\author[UMichECE]{Eron Ristich} 
\author[NEU]{Eduardo Sontag} 
\author[UMichECE]{Necmiye Ozay} 

\address[UMichECE]{Department of Electrical Engineering and Computer Science,
   University of Michigan, Ann Arbor, MI 48109 USA \\
   (e-mail: \{eristich, necmiye\}@umich.edu)}

\address[NEU]{Department of Electrical and Computer Engineering, and Department of Bioengineering, \\
   Northeastern University, Boston, MA 02115 USA \\
   (e-mail: e.sontag@northeastern.edu)}

\begin{abstract}                % Abstract of 50--100 words
Understanding when linear immersions of nonlinear dynamical systems exist is important since such immersions allow us to leverage the rich tools of linear system theory to analyze nonlinear dynamics. Recently, \cite{liu_non-existence_2023} showed that continuous-time dynamical systems that admit countably many but more than one $\omega$-limit sets cannot be immersed into finite dimensional linear systems with a one-to-one and continuous mapping. In this paper, we extend these results to discrete-time dynamics and show that similar obstructions exist also in discrete time. We further consider a generalization involving $\alpha$-limit sets. Several examples are provided to demonstrate the results.
\end{abstract}


\end{frontmatter}
%===============================================================================
%% NOTES:
%% There are a number of predefined theorem-like environments in
%% ifacconf.cls:
%%
%% \begin{thm} ... \end{thm}            % Theorem
%% \begin{lem} ... \end{lem}            % Lemma
%% \begin{claim} ... \end{claim}        % Claim
%% \begin{conj} ... \end{conj}          % Conjecture
%% \begin{cor} ... \end{cor}            % Corollary
%% \begin{fact} ... \end{fact}          % Fact
%% \begin{hypo} ... \end{hypo}          % Hypothesis
%% \begin{prop} ... \end{prop}          % Proposition
%% \begin{crit} ... \end{crit}          % Criterion

%% also, see the equation formating requirements:
%% use (\ref{eq:___}) instead of Eq.~\eqref{eq:___}

\section{Introduction}

The idea of embedding or immersing a complex nonlinear dynamical system into a higher-dimensional yet simpler one has important applications in analysis, design, and learning of nonlinear systems. Koopman operator theory formalizes this idea when the embedding dimension is infinite and embedding system is linear \citep{koopman1931hamiltonian,mauroy2020koopman,bevanda2021koopman}. However, in practice, one works with finite dimensional embeddings \citep{sontag1979polynomial,van1994locally,brunton2016koopman}. When such embeddings exist, one can devise algorithms to learn them or to directly infer properties of the underlying nonlinear system from data \citep{williams2015data,wang2023computation,haseli2021learning}.

Our work is inspired by the recent work by \cite{liu_non-existence_2023, liu_properties_2025, kvalheim2023linearizability}, which establishes necessary or sufficient conditions for (non)existence of such embeddings into finite dimensional linear systems in \emph{continuous-time}. Our goal is to extend some of these results to \emph{discrete-time} dynamical systems, which are arguably more common in the context of data-driven applications. In particular, we show that when a discrete-time dynamical system admits countably many but more than one $\omega$-limit sets or $\alpha$-limit sets and if it can be immersed into a finite dimensional linear system with a continuous mapping, this mapping collapses all limit sets; i.e., it cannot be one-to-one. Since most parametrizations (e.g., neural networks) used to learn the embedding functions are continuous, our results reveal an obstruction to obtaining one-to-one embedding functions that would allow the behavior of the underlying system to be uniquely recovered from that of the embedding system.

\emph{Notations:} We denote the closure of set $X$ by $\cl{X}$. The symbols $\mb{R}$ and $\mb{R}^+$ denote the real line and the set of non-negative real numbers. The symbols $\mb{Z}$ and $\mb{Z}^+$ denote the set of integers and the set of non-negative integers.

%\vspace{-7mm}
\section{Preliminaries}

\subsection{Discrete-time Dynamical Systems}
In this paper, we restrict our attention to a class of autonomous dynamical systems which have solutions given in discrete time by a sequence of states generated by an iterated function on a separable metric space $\mc{X}$. Specifically, given an initial state $x_0 \in \mc{X}$ and function $f: \mc{X} \rightarrow \mc{X}$,\footnote{Note that this implies that $\mc{X}$ is forward invariant under $f$.} we are interested in systems of the form
\begin{equation} \label{eq:system}
    x_{k+1} = f(x_k), \qquad x_k \in \mc{X}, k \in \mathbb{Z}^+,
\end{equation}
which have solutions given by the sequence $x_k$, defined by
\begin{equation} \label{eq:sys_solution}
    x_k = f^{k}(x_0), \qquad k \geq 0,
\end{equation}
where $f^{k}$ denotes the k-th iterate of $f$. Further, we denote the forward orbit generated by $f$ at $x_0$ as $\orb{f}{+}{x_0} := \{f^k(x_0) ~|~ k \in \mathbb{N}\}$, and the corresponding forward trajectory as the sequence $\seq{f}{+}{x_0} := (f^k(x_0))_{k \in \mathbb{N}}$. If $f$ is invertible as a map on $\mc{X}$,\footnote{That is, $f^{-1}: \mc{X} \rightarrow \mc{X}$, and $\mc{X}$ is backward invariant.} then in reverse time its solutions are denoted
\begin{equation}
    x_{-k} = f^{-k}(x_0), \qquad k > 0,
\end{equation}
where $f^{-k}$ denotes the $k$-th iterate of $f^{-1}$. Further, we define the backward orbit generated by $f$ at $x_0$ as $\orb{f}{-}{x_0} := \{f^{-k}(x_0) ~|~ k \in \mathbb{N}\}$, and the corresponding backward trajectory as the sequence $\seq{f}{-}{x_0} := (f^{-k}(x_0))_{k \in \mathbb{N}}$. 
% When $f$ is invertible as a map on $\mc{X}$, we denote the full orbit generated by $f$ at $x_0$ as $\orb{f}{}{x_0} := \{f^k(x_0) ~|~ k \in \mathbb{Z}\}$, and the corresponding trajectory is the bi-infinite sequence $\seq{f}{}{x_0} := (f^k(x_0))_{k \in \mathbb{Z}}$. When $f$ is not invertible, the full orbit (trajectory) is exactly the forward orbit (trajectory).

Throughout this paper, we assume that $f$ and its inverse, when it exists, are continuous\footnote{The continuity of $f$ guarantees that its limit sets are invariant (see \cite{alligood_chaos_2000}). That is, every point in any limit set maps under $f$ to another point in that same limit set. 
%A less restrictive condition is requiring $f$ and its inverse to be sequentially continuous in the neighborhood of its limit sets.
}. Additionally, we assume the following:
\begin{assum}[Path connectedness] \label{assum:path_connected}
    $\mc{X}$ is path connected.
\end{assum}
% Assuming path-connectedness of $\mc{X}$ excludes certain classes of chaotic systems which have fractal invariant sets. 
In general, as non-constant discrete-time trajectories have disconnected images, this assumption is more restrictive in discrete-time dynamical systems than continuous-time systems.

% \begin{rem}
%     If $\mc{X}$ is a connected manifold, then it is locally Euclidean, and thus path connected (\cite{lee_introduction_2011}).
% \end{rem}

\subsection{Definitions}
As in \cite{liu_non-existence_2023, liu_properties_2025}, we study nonlinear systems that contain multiple limit sets.

\begin{defn}[$\omega$-limit set] \label{def:omega_set}
    Given an initial state $\xi \in \mc{X}$, the $\omega$-limit set $\omega_{\mc{X}}(\xi)$ is the set of all $x \in \mc{X}$ such that there exists a sequence of positive indices $k_j \rightarrow \infty$ having $\lim_{j \rightarrow \infty} f^{k_j}(\xi) = x$ (\cite{hirsch_differential_2013}).
    % The $\omega$-limit set of an initial state $\xi \in \mc{X}$ is denoted by $\omega_{\mc{X}}(\xi)$. Following the definition from \cite{hirsch_differential_2013}, $\omega_{\mc{X}}(\xi)$ is the set of all $x \in \mc{X}$ such that there exists a sequence of indices $k_j \rightarrow \infty$ having $\lim_{j \rightarrow \infty} f^{k_j}(\xi) = x$.
\end{defn}

\begin{defn}[$\alpha$-limit set] \label{def:alpha_set}
    When $f$ is invertible as a map on $\mc{X}$, given an initial state $\xi \in \mc{X}$, the $\alpha$-limit set $\alpha_{\mc{X}}(\xi)$ is the set of all $x \in \mc{X}$ such that there exists a positive sequence of indices $k_j \rightarrow \infty$ having $\lim_{j \rightarrow \infty} f^{-k_j}(\xi) = x$. That is, $\alpha_{\mc{X}}(\xi)$ is the $\omega$-limit set of $\xi$ when time is reversed.
    % When $f$ is invertible as a map on $\mc{X}$, then the $\alpha$-limit set of an initial state $\xi \in \mc{X}$ is denoted by $\alpha_{\mc{X}}(\xi)$. Analogous to Def.~\ref{def:omega_set}, $\alpha_{\mc{X}}(\xi)$ is the set of all $x \in \mc{X}$ such that there exists a sequence of indices $k_j \rightarrow \infty$ having $\lim_{j \rightarrow \infty} f^{-k_j}(\xi) = x$. That is, $\alpha_{\mc{X}}(\xi)$ is the $\omega$-limit set of $\xi$ when time is reversed.
\end{defn}

When the domain of the $\alpha$ and $\omega$ limit sets is clear from context, no subscript is used (i.e., $\alpha(\cdot)$ and $\omega(\cdot)$).

\begin{defn}[Unbounded] \label{def:unbounded}
    For initial state $\xi \in \mc{X}$, the forward trajectory $\seq{f}{+}{\xi}$ is \textit{unbounded} if, for every bounded set $P \subseteq \mc{X}$, there exists positive $k$ such that $f^k(\xi) \notin P$. The forward trajectory is \textit{bounded} if it is not unbounded.
\end{defn}

The next lemma shows that if the trajectory of a linear system leaves every bounded set as in the above definition, it indeed leaves every bounded set permanently.

\begin{lem} \label{lma:linear_unboundedness}
    Suppose the system in \eqref{eq:system} is linear with $\mc{X} = \mathbb{R}^n$, having $x_{k+1} = f(x_k) := A x_{k}$ for some $A \in \mathbb{R}^{n \times n}$. For initial condition $\xi \in \mc{X}$, if $\seq{f}{+}{\xi}$ is unbounded, then $\lim_{k \rightarrow \infty} \|f^k(\xi)\| = \infty$, that is, for every bounded set $P \subseteq \mc{X}$, there exists $K \in \mathbb{Z}^+$ such that $f^k(\xi) \notin P$ for all $k \geq K$.
\end{lem}

\begin{pf}
As the system is linear, we can decompose $\xi$ into three vectors, $\xi_s$, $\xi_1$, and $\xi_u$, corresponding to the stable ($|\lambda| < 1$), unit ($|\lambda| = 1$ for trivial Jordan blocks), and unstable ($|\lambda| > 1$ or $|\lambda| = 1$ with non-trivial Jordan blocks) invariant subspaces of $A$. Importantly, as $A$ admits a Jordan canonical form $J$, the behavior of $A^k$ on each invariant subspace is determined by the corresponding Jordan blocks in $J^k$.

As $J$ is block diagonal, consider the $m \times m$ Jordan block $J_m$ corresponding to eigenvalue $\lambda$. Its $k$-th power is
\begin{equation} \label{eq:jordan_block}
    J_m^k = (\lambda I+N)^k = \sum_{i=0}^{\min(k,m-1)} \binom{k}{i}\lambda^{\,k-i}N^i,
\end{equation}
where $N \neq 0$ is a nilpotent matrix with ones on the superdiagonal such that $N^m = 0$. It follows that
\begin{align}
    \lim_{k \rightarrow \infty} \|A^k \xi_s\| &= 0, \\
    \lim_{k \rightarrow \infty} \|A^k \xi_1\| &= \|\xi_1\| < \infty, \\
    \forall k \geq N, \|A^k \xi_u\| &\geq c \alpha^k k^j, 
\end{align}
for some $N \in \mathbb{Z}^+$, real constants $c > 0$, $\alpha \geq 1$, and integer $j \geq 0$. In the case of $\xi_u$, if there exists $\lambda = 1$ with non-trivial Jordan block, then $\alpha = 1$. In this case, if $\xi_u$ linearly depends on generalized eigenvectors of $\lambda = 1$ with rank greater than 1, then $j$ is strictly larger than 0. Otherwise, if $\xi_u$ linearly depends only on generalized eigenvectors of $\lambda = 1$ with rank 1, then $\|A^k \xi_u\| = \|\xi_u\| < \infty$. This implies that the only unbounded term is $\xi_u$, whose norm approaches infinity as $k \rightarrow \infty$. {\hfill ~}\qed
\end{pf}

\begin{defn}[Precompact] \label{def:Precompact}
    The set $S \subseteq \mc{X}$ is called \textit{precompact} if $\cl{S}$ is compact with respect to the topology on $\mc{X}$. In particular, given initial state $\xi \in \mc{X}$, the trajectory through $\xi$ is called \textit{forward precompact} if $\orb{f}{+}{\xi}$ is precompact. Similarly, the trajectory through $\xi$ is called \textit{backward precompact} if $\orb{f}{-}{\xi}$ is precompact.
\end{defn}

\begin{lem} \label{lma:omega_nonempty}
    For any $\xi \in \mc{X}$, the limit set $\omega(\xi)$ is nonempty if the trajectory through $\xi$ is forward precompact in $\mc{X}$. If the system is linear with $\mc{X} = \mathbb{R}^n$, then the converse is also true.
\end{lem}
\begin{pf}
Suppose the trajectory through $\xi$ is forward precompact in $\mc{X}$. Then by definition, there exists a subsequence $k_j$ of $k \rightarrow \infty$ such that $f^{k_j}(\xi)$ converges to a point $x$ in the closure of $\orb{f}{+}{\xi}$. Thus, $x \in \omega(\xi)$ and $\omega(\xi)$ is nonempty.

Now suppose that the system in \eqref{eq:system} is linear, i.e., $x_{k+1} = A x_{k}$ for some $A \in \mathbb{R}^{n \times n}$. Suppose $\omega(\xi)$ is nonempty. If the forward trajectory $\seq{f}{+}{\xi}$ is unbounded, then we know that, by Lemma~\ref{lma:linear_unboundedness}, $\lim_{k \rightarrow \infty} \|f^k(\xi)\| = \infty$. It follows that any subsequence $k_j$ has the same limit, and $\omega(\xi)$ is empty. As $\omega(\xi)$ is assumed to be nonempty, the forward trajectory $\seq{f}{+}{\xi}$ must be bounded, and by the Bolzano-Weierstrass theorem, the closure of $\orb{f}{+}{\xi}$ in $\mathbb{R}^n$ is compact, hence the trajectory through $\xi$ is forward precompact in $\mathbb{R}^n = \mc{X}$. {\hfill ~}\qed
\end{pf}

\begin{defn}[Domain of attraction/repulsion]
    Let $\mc{W}$ be the set of all $\omega$-limit sets of the system in \eqref{eq:system}, and $\mc{A}$ be the set of all $\alpha$-limit sets. For each $\Omega \in \mc{W}$, its domain of attraction on $\mc{X}$ is
    \begin{equation}
        D^+_{\mc{X}}(\Omega) = \{\xi \in \mc{X} ~|~ \omega(\xi) = \Omega\}.
    \end{equation}
    For each $\Gamma \in \mc{A}$, its domain of repulsion on $\mc{X}$ is
    \begin{equation}
        D^-_{\mc{X}}(\Gamma) = \{\xi \in \mc{X} ~|~ \alpha(\xi) = \Gamma\}.
    \end{equation}
\end{defn}

\begin{rem}
    Consider any $\Omega \in \mc{W}$. If $\xi \in D^+_{\mc{X}}(\Omega)$, then $\orb{f}{+}{\xi} \subseteq D^+_{\mc{X}}(\Omega)$, that is, $D^+_{\mc{X}}(\Omega)$ is forward invariant. Similarly, for any $\Gamma \in \mc{A}$, then $D^-_{\mc{X}}(\Gamma)$ is backward invariant.
\end{rem}

\begin{defn}[Closed basins]
    The system in \eqref{eq:system} has \textit{closed basins} if $D^+_{\mc{X}}(\Omega)$ is closed for all $\omega$-limit sets $\Omega \in \mc{W}$.
\end{defn}

\begin{lem} \label{lma:linear_closed_basins}
    Every linear system $x_{k+1} = A x_k$ with $A \in \mathbb{R}^{n \times n}$ and $\mc{X} = \mathbb{R}^n$ has closed basins. 
\end{lem}
\begin{pf}
Let $x_\infty$ be a limit point of $D^+_{\mc{X}}(\Omega)$ for some $\Omega \in \mc{W}$. Denote the span of $D^+_{\mc{X}}(\Omega)$ as $S$. Since $S$ contains $D^+_{\mc{X}}(\Omega)$ and is closed, $x_\infty \in S$. By superposition, as $D^+_{\mc{X}}(\Omega)$ is forward invariant, $S$ is also. 

By Lemma~\ref{lma:omega_nonempty}, all trajectories in $D^+_{\mc{X}}(\Omega)$ are forward precompact. By superposition, all trajectories in $S$ are also forward precompact. This means that the system restricted to $S$ is stable in the sense of Lyapunov, and thus there exists $M > 0$ such that $\|A^k x\| \leq M \|x\|$ for all $x \in S$ and $k \geq 0$. As $x_\infty$ is a limit point of $D^+_{\mc{X}}(\Omega)$, there exists a sequence $x_j$ in $D^+_{\mc{X}}(\Omega)$ such that $x_j \rightarrow x_\infty$. Then, for all $j > 0$, and for all $k \geq 0$,
\begin{equation}
\begin{split}
    \|f^k(x_\infty) - f^k(x_j)\| &= \|A^k(x_\infty - x_j)\| \\
    &\leq M\|x_\infty - x_j\| \xrightarrow[]{j \rightarrow \infty} 0.
\end{split}
\end{equation}
It follows that $\omega(x_\infty) = \omega(x_j) = \Omega$. Thus $x_\infty \in D^+_{\mc{X}}(\Omega)$, and we conclude that $D^+_{\mc{X}}(\Omega)$ is closed. {\hfill ~}\qed
\end{pf}

\begin{rem}
    When $A$ is invertible, then the time-reversed system $x_{k+1} = A^{-1} x_k$ also has closed basins, that is, $D^-_{\mc{X}}(\Gamma)$ is closed for all $\alpha$-limit sets $\Gamma \in \mc{A}$ of the original system.
\end{rem}

\begin{defn}[Immersion]
    The system in \eqref{eq:system} is \textit{immersed} in a system $z_{k+1} = g(z_k)$ on separable metric space $\mc{Z}$ if there exists a mapping $F: \mc{X} \rightarrow \mc{Z}$ satisfying
    \begin{equation}
        F \circ f = g \circ F.
    \end{equation}
    % The system in \eqref{eq:system} is \textit{immersed} in a system $z_{k+1} = g(z_k)$ on manifold $\mc{Z}$ if there exists a mapping $F: \mc{X} \rightarrow \mc{Z}$ such that for all initial states $\xi \in \mc{X}$ and all time $k \rightarrow \infty$,
    % \begin{equation}
    %     F(f^k(\xi)) = g^k(F(\xi)),
    % \end{equation}
    % where $g^k(F(\xi))$ is the solution to $z_{k+1} = g(z_k)$.
\end{defn}

\begin{rem}
    The term ``immersion'' is also frequently used in the study of maps between differentiable manifolds (see e.g. \cite{gallot_riemannian_2004}), and is unrelated to the notion of immersion between dynamical systems used throughout this work.
\end{rem}

\section{Main theorems}
% Thm. 1: Discrete-time analog of Zexiang's IFAC WC 2023 paper (for omega-limit sets and alpha-limit sets)
\begin{thm} \label{thm:super_theorem}
    Let Assumption \ref{assum:path_connected} hold, and assume that $f$ is invertible. Suppose further that
    \begin{enumerate}[
        leftmargin=*,
        label={(T\arabic*)}
    ]
        \item \label{cond:t1} $x_{k+1} = f(x_k)$ on $\mc{X}$ can be immersed by a continuous mapping $F$ in a system $z_{k+1} = g(z_k)$ on $\mc{Z}$, $g$ invertible and continuous with continuous inverse, where the lifted system and its time reversal both have closed basins;

        \item \label{cond:t2} Each trajectory of $x_{k+1} = f(x_k)$ on $\mc{X}$ is either forward precompact or backward precompact in $\mc{X}$;

        \item \label{cond:t3} The set $\mc{W} \cup \mc{A}$ is countable.
    \end{enumerate}
    Then, the set $\{F(S) ~|~ S \in \mc{W} \cup \mc{A}\}$ has exactly one maximal element, that is, there exists $S_{\max} \in \mc{W} \cup \mc{A}$ such that $F(S_{\max}) \supseteq F(S)$ for all $S \in \mc{W} \cup \mc{A}$. Further, if
    \begin{enumerate}[
        leftmargin=*,
        label={(T\arabic*)}
    ]
        \setcounter{enumi}{3}
        \item \label{cond:t4} For every $\Omega \in \mc{W}$ and $\Gamma \in \mc{A}$, there exists at least one forward precompact trajectory in $D^+_{\mc{X}}(\Omega)$ and backward precompact trajectory in $D^-_{\mc{X}}(\Gamma)$,
    \end{enumerate}
    then $\{F(S) ~|~ S \in \mc{W} \cup \mc{A}\}$ has exactly one element.
\end{thm}

That is to say, when conditions \ref{cond:t1}--\ref{cond:t4} hold, the immersion $F$ necessarily collapses all $\omega$-limit and $\alpha$-limit sets. When only conditions \ref{cond:t1}--\ref{cond:t3} hold, then Theorem \ref{thm:super_theorem} suggests that the immersion $F$ cannot fully distinguish the limit sets in $\mc{X}$, even when the limit sets of the original system are disjoint. 
% For example, suppose $\xi_1 \in \mc{X}$ and $\xi_2 \in \mc{X}$ have $\omega(\xi_1) \neq \omega(\xi_2)$, but $F(\omega(\xi_1)) \subseteq F(\omega(\xi_2))$. In this case, one cannot distinguish which $\omega$-limit set in the original system the trajectory $g^k(F(\xi_1))$ converges to.

\begin{rem}
    Condition \ref{cond:t2} requires that $\mc{X}$ does not include trajectories which, loosely speaking, come from and go to infinity.
\end{rem}

\begin{rem}
    Condition \ref{cond:t4} does not immediately follow from conditions \ref{cond:t1}--\ref{cond:t3}. For example, one can construct a discrete-time trajectory which is backward precompact but in forward time is unbounded as in Definition \ref{def:unbounded}, but returns to and accumulates in a compact region, causing it to have an $\omega$-limit set. However, if the only $\omega$-limit and $\alpha$-limit sets are equilibria or periodic orbits, then \ref{cond:t4} is trivially satisfied, as, e.g., in the former case, the constant trajectory contained by the limit set is both forward and backward precompact.
\end{rem}

\begin{cor} \label{cor:super_corollary}
    Suppose that conditions \ref{cond:t1}--\ref{cond:t4} hold, and $F$ is one-to-one. Then, $\mc{W} \cup \mc{A}$ has exactly one element.
\end{cor}

Corollary \ref{cor:super_corollary} implies that the original system has at most one $\omega$-limit or $\alpha$-limit set. A direct consequence of this, combined with Lemma \ref{lma:linear_closed_basins}, is the following corollary, which is helpful for proving the nonexistence of one-to-one immersions.
\begin{cor} \label{cor:super_nonexistence}
    Suppose that conditions \ref{cond:t2}, \ref{cond:t3}, and \ref{cond:t4} hold. If $\mc{X}$ contains more than one limit set, then, this system on $\mc{X}$ cannot be immersed by a one-to-one immersion $F$ in a system $z_{k+1} = g(z_k)$, having $g$ invertible and continuous, where the lifted system and its time reversal have closed basins (e.g., a linear system).
\end{cor}

Theorem \ref{thm:super_theorem} requires $f$ from the original system to be invertible. However, regardless of if $f$ is invertible, a similar result can be obtained for general discrete-time dynamical systems. Analogously to \cite{liu_non-existence_2023}, if we relax condition \ref{cond:t1} to only require the immersed system to have closed basins, and we tighten condition \ref{cond:t2} to require every trajectory to be forward precompact, then the following theorem also holds.

\begin{thm} \label{thm:main_theorem}
    Under Assumption \ref{assum:path_connected}, suppose further that
    \begin{enumerate}[
        leftmargin=*,
        label={(T\arabic*{$^\prime$})}
    ]
        \item \label{cond:t1p} $x_{k+1} = f(x_k)$ on $\mc{X}$ can be immersed by a continuous mapping $F$ in a system with closed basins $z_{k+1} = g(z_k)$ on $\mc{Z}$, having $g$ continuous;

        \item \label{cond:t2p} Every trajectory of $x_{k+1} = f(x_k)$ on $\mc{X}$ is forward precompact;

        \item \label{cond:t3p} The set $\mc{W}$ is countable.
    \end{enumerate}
    Then the set $\{F(\Omega) ~|~ \Omega \in \mc{W}\}$ has exactly one element.
\end{thm}

That is to say, when conditions \ref{cond:t1p}--\ref{cond:t3p} hold, the immersion $F$ necessarily collapses all $\omega$-limit sets. Similarly, if $f$ is invertible, the theorem also holds for the $\omega$-limit sets of the time reversed system, that is, the $\alpha$-limit sets of the original system.

When we restrict ourselves to one-to-one immersions, which we often seek to identify in practice, a direct corollary of Theorem \ref{thm:main_theorem} is the following.
\begin{cor} \label{cor:one_to_one}
    Suppose that conditions \ref{cond:t1p}--\ref{cond:t3p} hold, and $F$ is one-to-one, then $\mc{W}$ has exactly one element.
\end{cor}

Lastly, the discrete-time analogue to Corollary 3 of \cite{liu_non-existence_2023} is the following, which immediately follows from Theorem \ref{thm:main_theorem}, Corollary \ref{cor:one_to_one}, and Lemma \ref{lma:linear_closed_basins}.
\begin{cor} \label{cor:main_corollary}
    Suppose that conditions \ref{cond:t2p} and \ref{cond:t3p} hold. If $\mc{X}$ contains more than one $\omega$-limit set, then this system on $\mc{X}$ cannot be immersed by a continuous one-to-one immersion $F$ in a system with closed basins $z_{k+1} = g(z_k)$, having $g$ continuous.
\end{cor}

When $f$ is invertible as a map on $\mc{X}$, Theorem~\ref{thm:main_theorem} and Corollary~\ref{cor:main_corollary} also apply to the $\alpha$-limit sets of the system, as they are exactly the $\omega$-limit sets of the time reversed system.

\subsection{Proof of Theorem \ref{thm:super_theorem}}

To prove Theorem \ref{thm:super_theorem}, we first introduce the following lemmas revealing properties of immersions and of systems with closed basins. The proofs of these lemmas closely follow those of Lemmas 5, 6, 7, and 8 of \cite{liu_properties_2025}, but are adapted to the discrete-time setting. The assumption that allows the continuous-time proofs in \cite{liu_properties_2025} to be adapted to discrete-time is, in lieu of sufficient smoothness, the continuity of $f$, $g$, and their inverses, when they exist.

\begin{lem} \label{lma:omega_alpha_limit_equal}
    Suppose that $f$ is invertible, and the system $x_{k+1} = f(x_k)$ on $\mc{X}$ and its time reversal $x_{k+1} = f^{-1}(x_k)$ on $\mc{X}$ have closed basins. Then, for any $\omega$-limit set $\Omega \subseteq \mc{X}$ and $\alpha$-limit set $\Gamma \subseteq \mc{X}$ of the system, $D_{\mc{X}}^+(\Omega) \cap D_{\mc{X}}^-(\Gamma) \neq \emptyset$ implies that $\Omega = \Gamma$.
\end{lem}

\begin{pf}
    % Let $x \in D_{\mc{X}}^+(\Omega) \cap D_{\mc{X}}^-(\Gamma)$. Denote the orbit through $x$ as $P := \orb{f}{}{x}$.
    Let $x \in D_{\mc{X}}^+(\Omega) \cap D_{\mc{X}}^-(\Gamma)$. Denote the orbit through $x$ as $P := \{f^k(x) ~|~ k \in \mathbb{Z}\}$. Then by the definition of limit sets, $P \subseteq D_{\mc{X}}^+(\Omega) \cap D_{\mc{X}}^-(\Gamma)$ and $\Omega \cup \Gamma \subseteq \cl{P}$. In particular, as the system and its time reversal have closed basins, then $D_{\mc{X}}^+(\Omega)$ and $D_{\mc{X}}^-(\Gamma)$ are closed, and we have
    \begin{equation}
        \Omega \cup \Gamma \subseteq \cl{P} \subseteq D_{\mc{X}}^+(\Omega) \cap D_{\mc{X}}^-(\Gamma).
    \end{equation}
    Thus, $\Omega \subseteq D_{\mc{X}}^-(\Gamma)$ and $\Gamma \subseteq D_{\mc{X}}^+(\Omega)$. Limit sets are closed, and as $f$ and its inverse are continuous, the limit sets are also both forward and backward invariant (see \cite{alligood_chaos_2000}). As $\Omega$ is closed and invariant, and $\Omega \subseteq D_{\mc{X}}^-(\Gamma)$, then $\Gamma \subseteq \cl{\Omega} = \Omega$. Similarly, $\Omega \subseteq \cl{\Gamma} = \Gamma$. We conclude that $\Omega = \Gamma$. {\hfill ~}\qed
\end{pf}

\begin{lem} \label{lma:immersed_limit_set_equality}
    Let $F$ be a continuous immersion that maps $x_{k+1} = f(x_k)$ on $\mc{X}$ to $z_{k+1} = g(z_k)$ on $\mc{Z}$. For any $\xi \in \mc{X}$, if $\omega_{\mc{X}}(\xi)$ is non-empty, then $F(\omega_{\mc{X}}(\xi)) \subseteq \omega_{\mc{Z}}(F(\xi))$. Furthermore, if the trajectory through $\xi$ is forward precompact in $\mc{X}$, then $F(\omega_{\mc{X}}(\xi)) = \omega_{\mc{Z}}(F(\xi))$.
\end{lem}

\begin{pf}
    First, we show that $F(\omega_{\mc{X}}(\xi)) \subseteq \omega_{\mc{Z}}(F(\xi))$. Suppose $p \in \omega_{\mc{X}}(\xi)$. Then, pick a sequence of indices $k_i \rightarrow \infty$ such that $f^{k_i}(\xi) \rightarrow p$ as $k_i \rightarrow \infty$. Then, by the definition of an immersion and the continuity of $F$, $g^{k_i}(F(\xi)) = F(f^{k_i}( \xi))$ has the limit $F(p) \in \omega_{\mc{Z}}(F(\xi))$.

    Conversely, suppose that $q^\prime \in \omega_{\mc{Z}}(F(\xi))$, and pick a sequence of indices $k_i \rightarrow \infty$ so that $F(f^{k_i}(\xi)) = g^{k_i}(F(\xi)) \rightarrow q^\prime \in \mc{Z}$. As $(f^{k_i}(\xi))$ is a subsequence of the forward precompact trajectory through $\xi$ in $\mc{X}$, then as a set it is itself precompact, and there is a subsequence $k_{i_j}$ of $k_i$ and $p \in \omega_{\mc{X}}(\xi)$ such that $f^{k_{i_j}}(\xi) \rightarrow p$. Then, as $k_{i_j}$ is a subsequence of the convergent sequence $k_i$, we have $g^{k_{i_j}}(F(\xi)) \rightarrow q^\prime$. Moreover, $F(f^{k_{i_j}}(\xi)) \rightarrow F(p)$ by the continuity of $F$. Finally, by the definition of immersion, we have $g^{k_{i_j}}(F(\xi)) = F(f^{k_{i_j}}(\xi))$, and thus the two limits are equal, i.e., $q^\prime = F(p) \in F(\omega_{\mc{X}}(\xi))$. We conclude that $\omega_{\mc{Z}}(F(\xi)) \subseteq F(\omega_{\mc{X}}(\xi))$, and thus $F(\omega_{\mc{X}}(\xi)) = \omega_{\mc{Z}}(F(\xi))$.~{\hfill~}\qed
\end{pf}

When $f$ is invertible as a map on $\mc{X}$, Lemma \ref{lma:immersed_limit_set_equality} also applies to the $\alpha$-limit sets.

\begin{cor} \label{cor:immersed_limit_set_equality}
    If $\alpha_{\mc{X}}(\xi)$ is non-empty, then $F(\alpha_{\mc{X}}(\xi)) \subseteq \alpha_{\mc{Z}}(F(\xi))$. Furthermore, if the trajectory through $\xi$ is backward precompact in $\mc{X}$, then $F(\alpha_{\mc{X}}(\xi)) = \alpha_{\mc{Z}}(F(\xi))$.
\end{cor}

Lemma \ref{lma:immersed_limit_set_equality} and Corollary \ref{cor:immersed_limit_set_equality} show the equivalence of the limit sets of the lifted system and the lifted limit sets of the original system. Although this equivalence does not hold in general for the domains of attraction or repulsion, the following corollaries show that they are still related.

\begin{cor} \label{cor:immersed_domain_subset}
    For any $\xi \in \mc{X}$, if the trajectory through $\xi$ is forward precompact in $\mc{X}$, then $F(\xi) \in D^+_{\mc{Z}}(F(\omega_{\mc{X}}(\xi)))$
\end{cor}

Similarly, when $f$ is invertible as a map on $\mc{X}$, Corollary \ref{cor:immersed_domain_subset} also applies to the $\alpha$-limit sets.

% \begin{lem} \label{lma:immersed_domain_subset}
%     Let $F$ be a continuous immersion that maps $x_{k+1} = f(x_k)$ on $\mc{X}$ to $z_{k+1} = g(z_k)$ on $\mc{Z}$. For any $\xi \in \mc{X}$, if the trajectory $f^{k}(\xi)$ is forward precompact in $\mc{X}$, then $F(D^+_{\mc{X}}(\omega_{\mc{X}}(\xi))) \subseteq D^+_{\mc{Z}}(F(\omega_{\mc{X}}(\xi)))$.
% \end{lem}

% \begin{pf}
%     As $f^k(\xi)$ is forward precompact, then by Lemmas \ref{lma:omega_nonempty} and \ref{lma:immersed_limit_set_equality}, $\Omega := \omega_{\mc{X}}(\xi)$ is nonempty, and $F(\Omega) = \omega_{\mc{Z}}(F(\xi))$. Consider any $x \in D_{\mc{X}}^+(\Omega)$. By definition, there exists a sequence $k_i$ such that $k_i \rightarrow \infty$ and $f^{k_i}(x) \rightarrow x^\prime$ for some $x^\prime \in \Omega$. Then,
%     \begin{equation}
%         \begin{split}
%         \lim_{i \rightarrow \infty} g^{k_i}(F(x)) &= \lim_{i \rightarrow \infty} F(f^{k_i}(x)) \\
%         &= F(x^\prime) \in F(\Omega),
%         \end{split}
%     \end{equation}
%     where the first equality is true by the definition of immersions, and the second equality is true by the continuity of $F$. Hence, $F(x) \in D^+_{\mc{Z}}(F(\omega_{\mc{X}}(\xi)))$, and we conclude that $F(D^+_{\mc{X}}(\omega_{\mc{X}}(\xi))) \subseteq D^+_{\mc{Z}}(F(\omega_{\mc{X}}(\xi)))$. {\hfill ~}\qed
% \end{pf}

% When $f$ is invertible as a map on $\mc{X}$, Lemma \ref{lma:immersed_domain_subset} also applies to the $\alpha$-limit sets, as restricted to a backward invariant subset of $\mc{X}$.

\begin{cor} \label{cor:alpha_immersed_domain_subset}
    For any $\xi \in \mc{X}$, if the trajectory through $\xi$ is backward precompact in $\mc{X}$, then $F(\xi) \in D^-_{\mc{Z}}(F(\alpha_{\mc{X}}(\xi)))$.
\end{cor}

The following lemma provides sufficient conditions for two distinct limit sets to have the same image under $F$, which is central to the proof of Theorem \ref{thm:super_theorem}.

\begin{lem} \label{lma:nonempty_domains}
    Let $\widehat{\mc{W}}$ be the set of $\omega$-limit sets $\Omega \in \mc{W}$ whose domain of attraction in $\mc{X}$ contains at least one trajectory forward precompact in $\mc{X}$. Further, let $\widehat{\mc{A}}$ be the set of $\alpha$-limit sets $\Gamma \in \mc{A}$ whose domain of repulsion in $\mc{X}$ contains at least one trajectory backward precompact in $\mc{X}$. Then, under conditions \ref{cond:t1}--\ref{cond:t3}, for any pair $S, S^\prime \in \widehat{\mc{W}} \cup \widehat{\mc{A}}$, their images under $F$ are equal.
\end{lem}

\begin{pf}
    For any limit set $S \in \widehat{\mc{W}} \cup \widehat{\mc{A}}$, define
    \begin{equation}
    \begin{split}
        D_{\mc{X}}(S) &= \begin{cases}
            D^+_{\mc{X}}(S) & S \in \widehat{\mc{W}} \setminus \widehat{\mc{A}} \\
            D^-_{\mc{X}}(S) & S \in \widehat{\mc{A}} \setminus \widehat{\mc{W}} \\
            D^+_{\mc{X}}(S) \cup D^-_{\mc{X}}(S) & S \in \widehat{\mc{W}} \cap \widehat{\mc{A}},
        \end{cases} \\
        D_{\mc{Z}}(S) &= \begin{cases}
            D^+_{\mc{Z}}(F(S)) & S \in \widehat{\mc{W}} \setminus \widehat{\mc{A}} \\
            D^-_{\mc{Z}}(F(S)) & S \in \widehat{\mc{A}} \setminus \widehat{\mc{W}} \\
            D^+_{\mc{Z}}(F(S)) \cup D^-_{\mc{Z}}(F(S)) & S \in \widehat{\mc{W}} \cap \widehat{\mc{A}}.
        \end{cases}
    \end{split}
    \end{equation}
    Then, for any $S, S^\prime \in \widehat{\mc{W}} \cup \widehat{\mc{A}}$ such that $F(S) \neq F(S^\prime)$, we have
    \begin{equation}
        \begin{split}
            D^+_{\mc{Z}}(F(S)) \cap D^+_{\mc{Z}}(F(S^\prime)) &= \emptyset, \\
            D^-_{\mc{Z}}(F(S)) \cap D^-_{\mc{Z}}(F(S^\prime)) &= \emptyset, \\
            D^+_{\mc{Z}}(F(S)) \cap D^-_{\mc{Z}}(F(S^\prime)) &= \emptyset, \\
            D^-_{\mc{Z}}(F(S)) \cap D^+_{\mc{Z}}(F(S^\prime)) &= \emptyset. \\
        \end{split}
    \end{equation}
    The first two equations are true by definition of the domain of attraction, while the last two follow by Lemma \ref{lma:omega_alpha_limit_equal} and the fact that the immersed system has closed basins. As such, if $F(S) \neq F(S^\prime)$ then $D_{\mc{Z}}(S)$ and $D_{\mc{Z}}(S^\prime)$ are disjoint.

    By condition \ref{cond:t2}, every trajectory is either forward or backward precompact. Thus, for any $\xi \in \mc{X}$, there exists $S_i \in \widehat{\mc{W}} \cup \widehat{\mc{A}}$ such that $\xi \in D_{\mc{X}}(S_i)$, and $i \in I$, a countable set by condition \ref{cond:t3}. As $\xi$ was arbitrary, we have
    \begin{equation} \label{eq:super_F_partition}
        F(\mc{X}) = \bigcup_{i \in I} \left(D_{\mc{Z}}(F(S_i)) \cap F(\mc{X})\right),
    \end{equation}
    where inclusion follows from Corollaries \ref{cor:immersed_domain_subset} and \ref{cor:alpha_immersed_domain_subset}, and equality from the intersection with $F(\mc{X})$.
    % Further, by Lemmas \ref{lma:immersed_limit_set_equality} and \ref{lma:immersed_domain_subset}, and Corollaries \ref{cor:immersed_limit_set_equality} and \ref{cor:immersed_domain_subset}, we have
    % \begin{equation} \label{eq:super_F_partition}
    % \begin{split}
    %     F(\mc{X}) &= \bigcup_{i \in I} F(D_{\mc{X}}(S_i)) \\
    %     &= \bigcup_{i \in I} (D_{\mc{Z}}(F(S_i)) \cap F(\mc{X})).
    % \end{split}
    % \end{equation}
    By condition \ref{cond:t1}, $D_{\mc{Z}}(F(S_i))$ is closed, and thus its intersection with $F(\mc{X})$ is closed with respect to the subspace topology induced on $F(\mc{X})$, for all $i \in I$. Importantly, $D_{\mc{Z}}(F(S_i)) \cap F(\mc{X})$ must be nonempty, as $D_{\mc{Z}}(F(S_i))$ is closed and the set of its limit points, $F(S_i) \subseteq F(\mc{X})$, is nonempty.
    
    Suppose there exist $S$ and $S^\prime$ such that $F(S) \neq F(S^\prime)$. Then, by Equation \eqref{eq:super_F_partition}, $F(\mc{X})$ is a disjoint union of a countable collection of closed sets. By Assumption \ref{assum:path_connected}, $\mc{X}$ is path connected, and as $F$ is continuous, $F(\mc{X})$ is path connected as well. Thus, by the extension to a theorem from \cite{sierpinski_theoreme_1918} given by \cite{liu_properties_2025}, only one of the sets in $\{D_{\mc{Z}}(F(S_i)) \cap F(\mc{X})\}_{i \in I}$ can be nonempty, contradicting that $D_{\mc{Z}}(F(S_i)) \cap F(\mc{X})$ is nonempty for all $i \in I$. As such, $F(S) = F(S^\prime)$ for all $S, S^\prime \in \widehat{\mc{W}} \cup \widehat{\mc{A}}$. {\hfill ~}\qed
\end{pf}

\begin{pf*}{Proof of Theorem \ref{thm:super_theorem}.}
    Condition \ref{cond:t2} and Lemma \ref{lma:omega_nonempty} guarantee that $\widehat{\mc{W}} \cup \widehat{\mc{A}}$ is nonempty. Take any element $S_{\max} \in \widehat{\mc{W}} \cup \widehat{\mc{A}}$, and consider any $S \in \mc{W} \cup \mc{A}$. We show that $F(S) \subseteq F(S_{\max})$. There are three cases.
    \begin{enumerate}
        \item $S \in \widehat{\mc{W}} \cup \widehat{\mc{A}}$. Then, by Lemma \ref{lma:nonempty_domains}, $F(S) = F(S_{\max})$.

        \item $S \in \mc{W} \setminus (\widehat{\mc{W}} \cup \widehat{\mc{A}})$. In this case, pick any $\xi \in D_{\mc{X}}^+(S)$. As $S \notin \widehat{\mc{W}} \cup \widehat{\mc{A}}$, then the trajectory through $\xi$ must not be forward precompact in $\mc{X}$, and by condition \ref{cond:t2}, must instead be backward precompact. By Lemma \ref{lma:omega_nonempty}, $\alpha_{\mc{X}}(\xi)$ exists, and belongs to $\widehat{\mc{A}}$. It follows from Corollary \ref{cor:immersed_limit_set_equality} and Lemma \ref{lma:nonempty_domains} that $F(\alpha_{\mc{X}}(\xi)) = \alpha_{\mc{Z}}(F(\xi)) = F(S_{\max})$. Further, we have by Lemma \ref{lma:immersed_limit_set_equality} that $\omega_{\mc{Z}}(F(\xi))$ is nonempty and contains $F(S)$. Lastly, as $\xi$ has nonempty $\omega$-limit and $\alpha$-limit sets, we have that $F(\xi) \in D_{\mc{Z}}^+(\omega_{\mc{Z}}(F(\xi))) \cap D_{\mc{Z}}^-(\alpha_{\mc{Z}}(F(\xi)))$, and Lemma \ref{lma:omega_alpha_limit_equal} implies that $\omega_{\mc{Z}}(F(\xi)) = \alpha_{\mc{Z}} (F(\xi))$. Thus,
        \begin{equation}
            F(S) \subseteq \omega_{\mc{Z}} (F(\xi)) = \alpha_{\mc{Z}} (F(\xi)) = F(S_{\max}).
        \end{equation}

        \item $S \in \mc{A} \setminus (\widehat{\mc{W}} \cup \widehat{\mc{A}})$. Applying the same argument as was applied in Case 2, one can show that $F(S) \subseteq F(S_{\max})$
    \end{enumerate}
    Lastly, under condition \ref{cond:t4}, we have $\mc{W} \cup \mc{A} = 
    \widehat{\mc{W}} \cup \widehat{\mc{A}}$, and by Lemma \ref{lma:nonempty_domains}, $F(S) = F(S_{\max})$ for any $S \in \mc{W} \cup \mc{A}$. We conclude that, in this case, $\{F(S) ~|~ S \in \mc{W} \cup \mc{A}\}$ has exactly one element. {\hfill ~}\qed 
\end{pf*}

\subsection{Proof of Theorem \ref{thm:main_theorem}}

Using a similar argument as the proof for Lemma \ref{lma:nonempty_domains} as well as Lemmas \ref{lma:immersed_limit_set_equality} and \ref{cor:immersed_domain_subset}, we are able to prove Theorem \ref{thm:main_theorem} directly. The proof follows similarly to the proof of the main theorem of \cite{liu_non-existence_2023}, adapted to the discrete-time setting.

\begin{pf*}{Proof of Theorem \ref{thm:main_theorem}.}
    By condition \ref{cond:t2p}, every trajectory is forward precompact. Thus, for any $\xi \in \mc{X}$, there exists $S_i \in \mc{W}$ such that $\xi \in D_{\mc{X}}^+(S_i)$, and $i \in I$, a countable set by condition \ref{cond:t3p}. As $\xi$ was arbitrary, we have
    \begin{equation} \label{eq:main_F_partition}
        F(\mc{X}) = \bigcup_{i \in I} \left(
            D^+_{\mc{Z}} (F(\Omega_i)) \cap F(\mc{X})
        \right),
    \end{equation}
    where inclusion follows from Corollary \ref{cor:immersed_domain_subset}, and equality from the intersection with $F(\mc{X})$. By condition \ref{cond:t1p}, $D^+_{\mc{Z}} (F(\Omega_i))$ is closed, and thus $D^+_{\mc{Z}} (F(\Omega_i)) \cap F(\mc{X})$ is closed in the subspace topology induced on $F(\mc{X})$, for all $i \in I$. Importantly, $D^+_{\mc{Z}}(F(\Omega_i)) \cap F(\mc{X})$ must be nonempty, as $D^+_{\mc{Z}}(F(\Omega_i))$ is closed and the set of its limit points, $F(\Omega_i) \subseteq F(\mc{X})$, is nonempty.
    
    Domains of attraction are disjoint, and as such, for $\omega$-limit sets $\Omega, \Omega^\prime \in \mc{W}$, either $F(\Omega) = F(\Omega^\prime)$, or $D^+_{\mc{Z}} (F(\Omega))$ and $D^+_{\mc{Z}} (F(\Omega^\prime))$ are disjoint. Suppose there exist $\Omega$ and $\Omega^\prime$ such that the second case is true. Then, by Equation \eqref{eq:main_F_partition}, $F(\mc{X})$ is a disjoint union of a countable collection of closed sets. By Assumption \ref{assum:path_connected}, $\mc{X}$ is path connected, and as $F$ is continuous, $F(\mc{X})$ is path connected as well. Thus, by the extension to a theorem from \cite{sierpinski_theoreme_1918} given by \cite{liu_non-existence_2023}, only one of the sets in $\{D^+_{\mc{Z}} (F(\Omega_i)) \cap F(\mc{X})\}_{i \in I}$ can be nonempty, contradicting that $D^+_{\mc{Z}} (F(\Omega_i)) \cap F(\mc{X})$ is nonempty for all $i \in I$. As such, $F(\Omega) = F(\Omega^\prime)$ for all $\Omega, \Omega^\prime \in \mc{W}$. {\hfill ~}\qed
\end{pf*}

\section{Examples}
% Can we come up with an example that has multiple omega-limit sets AND an unbounded trajectory of the system?

\begin{exmp} \label{ex:example_1}
Consider the 1-dimensional system
\begin{equation} \label{eq:example_1_sys}
    x_{k+1} = -\frac{3 x_k - 1}{x_k - 3}.
\end{equation}
There are two $\omega$-limit sets of this system, namely $\{-1\}$ and $\{1\}$. Let $\mc{X} = (-\infty, 1)$, which only contains one $\omega$-limit set, $\{-1\}$. We can map the system in \eqref{eq:example_1_sys} to the linear system $z_{k+1} = (1/2) z_k$ with the one-to-one mapping
\begin{equation} \label{eq:example_1_immersion}
z = F(x) = \frac{x+1}{x-1}.
\end{equation}
Note that when the domain is extended to $\mc{X}^\prime = (-\infty, 1]$, that $F$ is not an immersion as $F(1)$ is not defined. This can be explained as all trajectories are precompact, but $\mc{X}^\prime$ contains two $\omega$-limit sets, so by Corollary \ref{cor:main_corollary}, $\mc{X}^\prime$ does not have a one-to-one linear immersion.
\end{exmp}

\begin{exmp} \label{ex:alpha_limit_example}
The inverse of the system in Example~\ref{ex:example_1} is
\begin{equation} \label{eq:example_1_sys_inverse}
    x_{k+1} = \frac{3 x_{k} + 1}{x_{k} + 3}.
\end{equation}
As the $\alpha$-limit sets of the system in \eqref{eq:example_1_sys} are exactly the $\omega$-limit sets of the inverse dynamics in \eqref{eq:example_1_sys_inverse}, then the system has two $\alpha$-limit sets, namely $\{-1\}$ and $\{1\}$. Consider the backward invariant set $\mc{X} = (-1, \infty)$, which only contains one $\alpha$-limit set, $\{1\}$. We can map the system in \eqref{eq:example_1_sys_inverse} to the linear system $z_{k+1} = (1/2) z_{k}$ with the one-to-one mapping
\begin{equation} \label{eq:example_1_map_inverse}
    z = F(x) = \frac{x - 1}{x + 1}.
\end{equation}
When the domain is extended to $\mc{X}^\prime = [-1,\infty)$, $F$ is not an immersion as $F(-1)$ is not defined. This can be explained as all trajectories are precompact, but $\mc{X}^\prime$ contains two $\alpha$-limit sets, so by Corollary \ref{cor:main_corollary}, $\mc{X}^\prime$ does not have a one-to-one linear immersion.
\end{exmp}

% \begin{rem}
%     The immersions in Examples \ref{ex:example_1} and \ref{ex:alpha_limit_example} are well defined with respect to the forward invariant domain of attraction for the $\omega$-limit set $\{-1\}$ and the backward invariant domain of repulsion for the $\alpha$-limit set $\{1\}$. 
% \end{rem}

\begin{exmp} \label{ex:example_2}
Consider the 1-dimensional system
\begin{equation} \label{eq:example_2_sys}
    x_{k+1} = 2 \cot^{-1}\left( \frac{1}{\sqrt{2}} \cot\left(\frac{x_k}{2}\right)\right),
\end{equation}
whose right-hand side continuously extends to the origin such that, by letting $\mc{X} = [0,\pi]$, the $\omega$-limit and $\alpha$-limit sets of the system are $\{0\}$ and $\{\pi\}$. If we define $y = F(x) = \cos(x)$, then the system in \eqref{eq:example_2_sys} can be immersed in the system
\begin{equation} \label{eq:example_2_good_immersion}
\begin{split}
    y_{k+1} &= -\frac{3 y_k - 1}{y_k - 3},
\end{split}
\end{equation}
with $|y| \leq 1$. That is, the system in Eqn.~\eqref{eq:example_2_sys} on $\mc{X}$ can be immersed in the system in Eqn.~\eqref{eq:example_1_sys} on $\mc{Z} = [-1,1]$. Notably, all trajectories are precompact, a one-to-one immersion exists, but $\mc{W}$ has two elements. By Theorem \ref{thm:main_theorem}, this is only possible if the system in Eqn.~\eqref{eq:example_2_good_immersion} does not have closed basins. Indeed, we can see that $D^+_{\mc{Z}}(\{-1\})$ is $[-1,1)$, an open set. If the linear immersion in Example \ref{ex:example_1} is used for $y$, we see that the immersion in Eqn.~\eqref{eq:example_1_immersion} is not defined at $x=0$. This can be explained as all trajectories are precompact in $\mc{X}$, but the interval $[0,\pi]$ contains two $\omega$-limit sets.
\end{exmp}

\begin{exmp} \label{ex:unit_circle_example}
\begin{figure}
    \centering
    \includegraphics[width=0.8\linewidth]{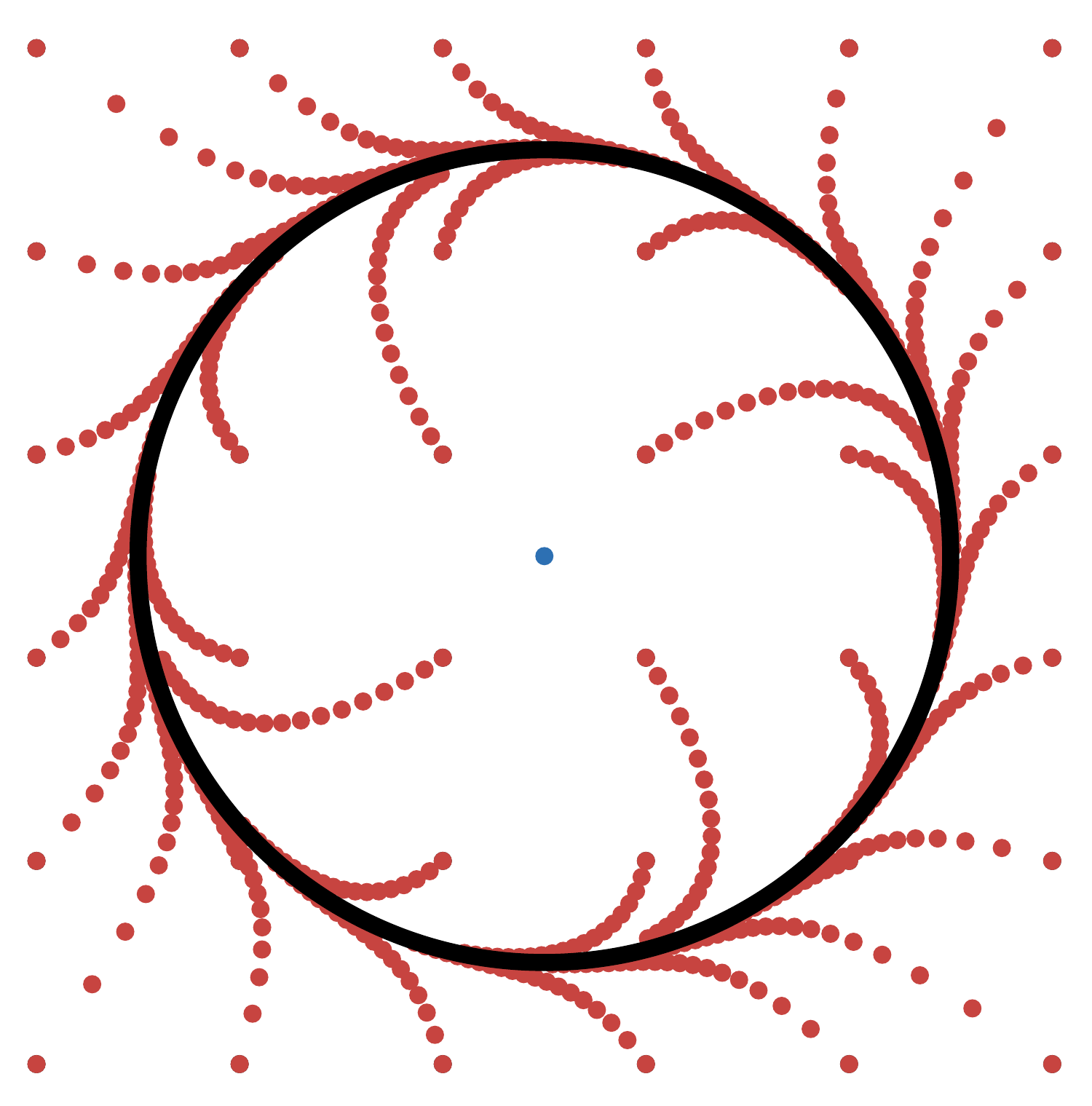}
    % \includesvg[width=0.8\linewidth]{figs/unit_circle_example.svg}
    \caption{Several sample trajectories of the two-dimensional dynamics in Example \ref{ex:unit_circle_example}. The non-trivial $\omega$-limit set, the unit circle, is shown in black. The isolated fixed point is shown in blue.}
    \label{fig:unit_circle_example}
\end{figure}

Consider the 2-dimensional system in $\mathbb{R}^2$
\begin{equation} \label{eq:unit_circle_sys}
    \begin{bmatrix}
        x^1_{k+1} \\ x^2_{k+1}
    \end{bmatrix} = \frac{2}{\|x_k\|_2 + 1} \begin{bmatrix}
        \cos{\theta} & \sin{\theta} \\ -\sin{\theta} & \cos{\theta}
    \end{bmatrix} \begin{bmatrix}
        x^1_k \\ x^2_k
    \end{bmatrix},
\end{equation}
having irrational parameter $\theta \in \mathbb{R} \setminus \mathbb{Q}$. One can think of this system as sequentially rotating and then scaling a two-dimensional vector. Due to the irrationality of $\theta$ (in degrees), the only two $\omega$-limit sets of this system are $\{0\}$ and the unit circle. See Fig.~\ref{fig:unit_circle_example}.

Let $\mc{X} = \mathbb{R}^2 \setminus \{0\}$. Define the one-to-one function $F: \mc{X} \rightarrow \mathbb{R}^3$ given by
\begin{equation}
    F(x) = \begin{bmatrix}
        \frac{x^1}{\|x\|_2} & \frac{x^2}{\|x\|_2} & \|x\|_2
    \end{bmatrix}^\intercal.
\end{equation}
Given a solution $x_k$ of Eqn.~\eqref{eq:unit_circle_sys}, it can be checked that $\begin{bmatrix}
    u_k & v_k & r_k
\end{bmatrix}^\intercal = F(x_k)$ is a solution to the system
\begin{align}
\begin{split}
\begin{bmatrix}
    u_{k+1} \\ v_{k+1}
\end{bmatrix} = \begin{bmatrix}
    \cos{\theta} & \sin{\theta} \\ -\sin{\theta} & \cos{\theta}
\end{bmatrix}
\begin{bmatrix}
    u_{k} \\ v_{k}
\end{bmatrix},\;
r_{k+1} = \frac{2 r_k}{r_k + 1}.
\end{split}
\end{align}
Similar to Example \ref{ex:example_1}, $r_k$ on $(0, \infty)$ can be immersed in a system $w_{k+1} = \frac{1}{2} w_k$ where we let
\begin{equation}
    w_k = \frac{r_k - 1}{r_k}.
\end{equation}
Thus, the system of $x$ in Eqn.~\eqref{eq:unit_circle_sys} on $\mc{X}$ can be immersed in the following linear system:
\begin{equation}
\begin{split}
    \begin{bmatrix}
        u_{k+1} \\ v_{k+1} \\ w_{k+1}
    \end{bmatrix} &= \begin{bmatrix}
        \cos{\theta} & \sin{\theta} & 0 \\ -\sin{\theta} & \cos{\theta} & 0 \\ 0 & 0 & 1/2
    \end{bmatrix} \begin{bmatrix}
        u_k \\ v_k \\ w_k
    \end{bmatrix},
\end{split}
\end{equation}
with a one-to-one immersion $\bar{F}$ 
\begin{equation}
    \begin{bmatrix}
        u & v & w
    \end{bmatrix}^\intercal = \bar{F}(x) = \begin{bmatrix}
        \frac{x^1}{\|x\|_2} & \frac{x^2}{\|x\|_2} & \frac{\|x\|_2 - 1}{\|x\|_2}
    \end{bmatrix}^\intercal.
\end{equation}
The immersion $\bar{F}$ is continuous in $\mc{X}$. However, if we extend $\mc{X}$ to $\mathbb{R}^2$, $\bar{F}$ is discontinuous at the origin, and thus not an immersion. This can be explained by Corollary \ref{cor:main_corollary}: as all trajectories are precompact and the set $\mathbb{R}^2$ contains two $\omega$-limit sets, there does not exist a one-to-one linear immersion for the system on $\mathbb{R}^2$.
\end{exmp}

\begin{rem}
    In Example \ref{ex:unit_circle_example}, if we instead choose $\theta$ to be rational, then there are uncountably many $\omega$-limit sets densely covering the unit circle, and Theorems \ref{thm:super_theorem} and \ref{thm:main_theorem} do not apply. In general, for discrete-time systems, such cases are easy to construct; for example, the system $x_{k+1} = -x_k$ is linear, but has uncountably many $\omega$-limit sets.
\end{rem}

\section{Conclusion}

In this work, we have shown that a discrete-time system with multiple limit sets cannot admit one-to-one continuous immersions to systems with closed basins under the conditions that \ref{cond:t2}, trajectories are precompact, and \ref{cond:t3}, there are at most countably many limit sets. More generally, under these conditions, this work suggests that continuous immersions necessarily collapse limit sets. These results show that the same topological obstructions which prevent continuous-time systems from admitting immersions also apply to discrete-time systems. 

\emph{Acknowledgements:} The authors would like to thank Zexiang Liu and Matthew Kvalheim for many constructive comments in a preliminary version, as well as the anonymous reviewers for their valuable feedback.

% \section*{DECLARATION OF GENERATIVE AI AND AI-ASSISTED TECHNOLOGIES IN THE WRITING PROCESS}
% During the preparation of this work the author(s) used [NAME TOOL / SERVICE] in order to [REASON]. After using this tool/service, the author(s) reviewed and edited the content as needed and take(s) full responsibility for the content of the publication.

\bibliography{refs}

\appendix

\end{document}